\documentclass[12pt]{article}

\usepackage{fullpage,float,amssymb,amsmath,amsthm}
\usepackage{thmtools}
\usepackage{thm-restate}
\usepackage{epsf}
\usepackage[usenames]{color}
\usepackage[colorlinks=true,
linkcolor=webgreen,
filecolor=webbrown,
citecolor=webgreen]{hyperref}

\definecolor{webgreen}{rgb}{0,.5,0}
\definecolor{webbrown}{rgb}{.6,0,0}

\usepackage{url}

\newcommand{\seqnum}[1]{\href{https://oeis.org/#1}{\underline{#1}}}

\newcommand{\N}{\mathbb{N}}
\newcommand{\R}{\mathbb{R}}
\newcommand{\lef}{\left}
\newcommand{\rig}{\right}
\newcommand{\fl}[1]{\ensuremath{\lef\lfloor #1 \rig\rfloor}}
\newcommand{\cl}[1]{\ensuremath{\lef\lceil #1 \rig\rceil}}
\newcommand{\dd}{\mathop{}\!d}
\renewcommand{\a}{\alpha}
\renewcommand{\b}{\beta}

\usepackage{color}
\usepackage{enumerate}

\author{Antonio Molina Lovett
and Jeffrey Shallit \\
School of Computer Science\\
University of Waterloo\\
Waterloo, ON  N2L 3G1 \\
Canada\\
{\tt ajmolina@uwaterloo.ca}\\
{\tt shallit@uwaterloo.ca} 
}

\title{Optimal Regular Expressions for Permutations}

\begin{document}

\maketitle

\theoremstyle{plain}
\newtheorem{theorem}{Theorem}
\newtheorem{corollary}[theorem]{Corollary}
\newtheorem{lemma}[theorem]{Lemma}
\newtheorem{proposition}[theorem]{Proposition}

\theoremstyle{definition}
\newtheorem{definition}[theorem]{Definition}
\newtheorem{example}[theorem]{Example}
\newtheorem{conjecture}[theorem]{Conjecture}

\theoremstyle{remark}
\newtheorem{remark}[theorem]{Remark}

\def\Que{{\mathbb{Q}}}
\def\Enn{{\mathbb{N}}}
\def\Zee{{\mathbb{Z}}}

\begin{abstract}
	The permutation language $P_n$ consists of all words that are permutations of a fixed alphabet of size $n$.  Using divide-and-conquer, we construct a regular expression $R_n$ that specifies $P_n$. We then give explicit bounds for the length of $R_n$, which we find to be $4^nn^{-(\lg n)/4+\Theta(1)}$, and use these bounds to show that $R_n$ has minimum size over all regular expressions specifying $P_n$.
\end{abstract}

\section{Introduction}

Given a regular language $L$ defined in some way, it is a challenging problem to find
good upper and lower bounds on the size of the smallest regular
expression specifying $L$.   (In this paper, by a regular expression, we always mean one using the operations of union, concatenation, and Kleene closure only.)  Indeed, as a computational problem,
it is known that determining the shortest regular expression corresponding to an NFA
is PSPACE-hard~\cite{Meyer&Stockmeyer:1972}.  Jiang and Ravikumar proved the analogous result for DFAs~\cite{Jiang&Ravikumar:1993}.  For more recent results on inapproximability, see~\cite{Gramlich&Schnitger:2007}.  

For nontrivial families of languages, only
a handful of results are already known.   For example, 
Ellul et al.~\cite{Ellul&Krawetz&Shallit&Wang:2005} showed that the shortest regular expression for the language
$\{ w \in \{ 0,1 \}^n \, : \, |w|_1 \text{ is even} \}$ is of length
$\Omega(n^2)$.    Here $|w|_1$ denotes the number of occurrences of the symbol $1$ in the word $w$. (A simple divide-and-conquer strategy provides a matching upper bound.) 
Chistikov et al.~\cite{Chistikov&Ivan&Lubiw&Shallit:2017}
showed that the regular language
$$ \{ ij \ : \ 1 \leq i < j \leq n \} $$
can be specified by a regular expression of size exactly $n (\lfloor \log_2 n \rfloor + 2) - 2^{\lfloor \log_2 n\rfloor + 1}$, and furthermore this bound is optimal.  Mousavi~\cite{Mousavi:2017} developed a general program for computing lower bounds on regular expression size for the binomial languages 
$$B(n,k) = \{ w \in \{ 0, 1 \}^n \, : \, |w|_1 = k \}.$$

Let $n$ be a positive integer, and define $\Sigma_n = \{ 1,2,\ldots, n \}$.  In this paper we study the finite language $P_n$ consisting of all permutations
of $\Sigma_n$.  Thus, for example,
$$ P_3 = \{ 123,132,213,231,312,321 \}.$$
We are interested in regular expressions that specify $P_n$.   In counting the length of regular expressions, we adopt the conventional measure of
{\it alphabetic length} (see, for example, \cite{Ehrenfeucht&Zeiger:1976}):  the length of a regular expression is the number of occurrences of symbols of the alphabet $\Sigma_n$.   Thus, other symbols, such as parentheses and {\tt +}, are ignored.

A brute-force solution, which consists of listing all the members of $P_n$ and separating them by the union symbol {\tt +}, evidently gives a regular expression for $P_n$ of alphabetic length $n \cdot n!$.
This can be improved to $n!\sum_{0 \leq i < n} 1/i! \sim e\cdot n!$
by tail recursion, where $E(S)$ represents a regular expression for all permutations of the symbols of $S$:
$$E (S) = \sum_{i \in S} i(E(S-\{ i \})); \quad E({i}) = i .$$
For example, for $P_4$ this gives\\
\centerline{
{\tt 1(2(34+43)+3(24+42)+4(23+32))+2(1(34+43)+3(14+41)+4(13+31))+}} \\
\centerline{
\quad \quad {\tt 3(1(24+42)+2(14+41)+4(12+21))+4(1(23+32)+2(13+31)+3(12+21))}.}
Can we do better?

Ellul et al.~\cite{Ellul&Krawetz&Shallit&Wang:2005} proved the following weak lower bound:   every regular expression for $P_n$ has alphabetic length at least $2^{n-1}$.  In this note we derive an upper bound through divide-and-conquer. We then show that the regular expression this strategy produces is, in fact, actually optimal.  This improves the result from \cite{Ellul&Krawetz&Shallit&Wang:2005}.  The language $P_n$ is of particular interest because its complement has short regular expressions, as shown in
\cite{Ellul&Krawetz&Shallit&Wang:2005}.  For other results concerning context-free grammars for $P_n$, see \cite{Ellul&Krawetz&Shallit&Wang:2005,Asveld:2006,Asveld:2008,Filmus:2011}.

\section{Divide-and-conquer}
\label{sec:divconq}

Consider the following divide-and-conquer strategy.  Let $S$ be an alphabet of cardinality $n$.   We consider all subsets $T \subseteq S$ of cardinality $\lfloor n/2 \rfloor$.  For each subset we recursively determine a regular expression for the permutations of $T$, a regular expression for the permutations of $S - T$, and concatenate them together.  This gives
\begin{equation}
	E(S) = \sum_{{T \subseteq S}\atop {|T| = \lfloor n/2 \rfloor}} (E(T))(E(S-T)); \quad E({i}) = i . \label{rec-defn}
\end{equation}
Finally, we define $R_n = E(\Sigma_n)$.

Thus, for example, we get\\[.1in]
\centerline{$R_4 =$ {\tt (12+21)(34+43)+(13+31)(24+42)+(23+32)(14+41)+}}\\
\centerline{{\tt (14+41)(23+32)+(24+42)(13+31)+(34+43)(12+21)}}\\
for $P_4$.

The alphabetic length of the resulting regular expression $R_n$ for all permutations of $\Sigma_n$ is then $f(n)$,
where 
$$
f(n) = \begin{cases}
    1, & \text{if $n = 1$}; \\[.1in]
    {\dbinom{n} {\lfloor n/2 \rfloor}} \bigl( f (\lfloor n/2 \rfloor) +
    f(\lceil n/2 \rceil) \bigr), & \text{if $n > 1$.}
    \end{cases}
$$
The first few values of $f(n)$ are given in the table below.
\begin{table}[H]
\begin{center}
\begin{tabular}{|c|c|}
\hline
$n$ & $f(n)$\\
\hline
1 & 1 \\
2 & 4 \\
3 & 15 \\
4 & 48 \\
5 & 190 \\
6 & 600 \\
7 & 2205 \\
8 & 6720 \\
9 & 29988 \\
10 & 95760 \\
\hline
\end{tabular}
\end{center}
\end{table}
\noindent It is sequence \seqnum{A320460} in the {\it On-Line Encyclopedia of Integer Sequences} \cite{Sloane:2018}.

It seems hard to determine a simple closed-form expression for $f(n)$.  Nevertheless, we can roughly estimate it as follows, at least when $n = 2^m$ is a power of $2$:
\begin{align*}
f(2^m) &= 2 {{2^m}\choose{2^{m-1}}} f(2^{m-1}) \\
&= 2^m {{2^m}\choose{2^{m-1}}} 
{{2^{m-1}}\choose{2^{m-2}}} \cdots 
{2 \choose 1}  \\
&= 2^m {{ (2^m)!} \over { (2^{m-1})! \, 
(2^{m-2})! \, \cdots \, 2! \,  1!}} .
\end{align*}
Substituting the Stirling approximation
$n! \sim  \sqrt{2 \pi n} (n/e)^n$ and
simplifying,
we get that $f(2^m)$ is roughly
equal to 
$$ 4^{2^m} e^{-1} \pi^{(1-m)/2} 2^{-(m^2-5m+6)/4}.$$
To make this precise, and make it work when $n$ is
not a power of $2$, however, takes more
work.

The rest of the paper is organized as follows: in Section~\ref{sec:opt}, we prove that our regular expression is in fact optimal, assuming one result that is proven at the end of the paper. In Section~\ref{analy}, we establish some inequalities related to Stirling's formula. In Section~\ref{bounds}, we connect these inequalities to $f(n)$ and obtain the estimate mentioned in the abstract.  Finally, in Section~\ref{sec:opt-choice} we use our obtained bounds on $f(n)$ to provide the missing piece in our optimality proof.

\section{Optimality}
\label{sec:opt}

In order to show that our regular expression has minimum possible length, we use the following property of $f(n)$ that we prove in Section~\ref{sec:opt-choice}:

\begin{restatable}{lemma}{optchoice}
\label{lem:opt-choice}
	If $n\ge1$, then every integer $0<k<n$ satisfies $\binom nk(f(k)+f(n-k))\ge f(n)$. Equality occurs if and only if $k=\fl{n/2}$ or $k=\cl{n/2}$.
\end{restatable}

For $n\ge1$ and $1\le k\le n!$, define $\ell(n,k)$ to be the minimum alphabetic length of a regular expression specifying a subset of $P_n$, where the subset has cardinality at least $k$.

\begin{lemma}
\label{lem:main-opt}
	If $n\ge1$ and $1\le k\le n!$, then $\ell(n,k)/k \ge \ell(n,n!)/n! \ge f(n)/n!$.
\end{lemma}
\begin{proof}
	We prove this by induction over the lexicographical ordering of pairs $(n,k)$. This is easy for our base case $n=k=1$, as the best regular expression is a single character. We thus suppose $n\ge2$.

	Consider a regular expression for a subset of $P_n$ of cardinality at least $k\geq 1$ that has minimum alphabetic length.   Clearly no such expression will involve $\epsilon$
	or $\emptyset$.
	We now consider the possibilities for the last (outermost) operation in the regular expression.  Clearly the only relevant possibilities are union and concatenation.

	If the last operation is a union, then it is the union of two subsets of $P_n$ of cardinalities $k_1,k_2\ge1$ where $k_1+k_2\ge k$, and $k_1,k_2<k$ by minimality. Then we get
	\begin{align*}
		\frac{\ell(n,k)}k
		&\ge \frac{\ell(n,k_1)+\ell(n,k_2)}{k_1+k_2} \\
		&\ge \min\lef\{\frac{\ell(n,k_1)}{k_1},\frac{\ell(n,k_2)}{k_2}\rig\} \\
		&\ge \frac{\ell(n,n!)}{n!}\ge\frac{f(n)}{n!}.
	\end{align*}
	
	If the last operation is a concatenation, then it is the concatenation of two regular expressions for subsets of $P_{n_1}$ and $P_{n_2}$ of cardinalities $k_1$ and $k_2$ respectively (possibly after changing alphabets) where $n_1+n_2=n$ and $k_1k_2\ge k$. By minimality, we have $n_1,n_2,k_1,k_2$ all positive, so $n_1,n_2<n$. We now obtain
	\begin{align*}
		\frac{\ell(n,k)}k
		&\ge \frac{\ell(n_1,k_1)+\ell(n_2,k_2)}{k_1k_2} \\
		&\ge \frac{\ell(n_1,n_1!)+\ell(n_2,k_2)}{n_1!\ k_2} \\
		&\ge \frac{\ell(n_1,n_1!)+\ell(n_2,n_2!)}{n_1!\ n_2!} \\
		&\ge \frac{f(n_1)+f(n_2)}{n_1!\ n_2!} \\
		&= \frac1{n!}\binom n{n_1}(f(n_1)+f(n-n_1)) \\
		&\ge \frac1{n!}\binom n{\fl{n/2}}(f(\fl{n/2})+f(\cl{n/2})) & \text{(by Lemma~\ref{lem:opt-choice})} \\
		&= \frac{f(n)}{n!}.
	\end{align*}
	In both cases, we get the desired inequalities for these choices of $n$ and $k$, completing our induction.
\end{proof}
\begin{theorem}
	Let $n\ge1$. Over all regular expressions for the permutation language $P_n$, the regular expression $R_n$ given by our divide-and-conquer strategy achieves the minimum alphabetic length.
\end{theorem}
\begin{proof}
	By our construction from Section~\ref{sec:divconq}, the regular expression $R_n$ specifies the entirety of $P_n$ and has alphabetic length $f(n)$. We thus get the upper bound $\ell(n,n!)\le f(n)$. By Lemma~\ref{lem:main-opt}, we have the matching lower bound $\ell(n,n!)\ge f(n)$. Thus, $R_n$ has minimum possible alphabetic length for a regular expression specifying $P_n$.
\end{proof}


\section{Analysis}

In what follows we use $\ln$ to denote the natural logarithm,
and $\lg$ to denote logarithms to the base $2$.

\label{analy}
	Define $S:\R_{>0}\to\R_{>0}$ to be the usual Stirling approximation~\cite{Robbins:1955}:
	\begin{displaymath}
		S(x) = \sqrt{2\pi x}(x/e)^x .
	\end{displaymath}

\begin{lemma}
\label{lem:sa}
    For every $x\ge1$, we have the bounds
    \begin{displaymath}
        S\bigl(x+\frac12\bigr)^2 \le S(x)S(x+1) \le e^{1/(2x)}S\bigl(x+\frac12\bigr)^2 .
    \end{displaymath}
\end{lemma}
\begin{proof}

	The first two derivatives of $\ln S(x)$ are
	\begin{align*}
		\frac\dd{\dd x}\ln S(x) &= \frac1{2x}+\ln x \\[.1in]
		\frac{\dd^2}{\dd x^2} \ln S(x) &= -\frac1{2x^2} + \frac1x,
	\end{align*}
	and so we see that $\ln S(x)$ is convex (that is, its derivative is increasing) for all $x>1/2$. Thus, by Jensen's inequality~\cite{Jensen:1906} and exponentiating, we obtain the lower bound
	\begin{displaymath}
		S(x)S(x+1) \ge S \bigl(x+\frac12 \bigr)^2
	\end{displaymath}
	for all $x\ge1$.
	
	Now, using the mean value theorem twice, we get
	that 
	\begin{align}
	    (\ln S(x)) + {1\over 2}\mu &\leq \ln S(x + {1 \over 2}) \label{ap1}\\
	    \ln S(x+1) &\leq (\ln S(x+{1 \over 2})) + {1 \over 2} M,
	    \label{ap2}
	\end{align}
	where 
	\begin{align*}
	\mu &= \inf_{z \in [x,x+{1 \over 2}]} (\ln S(z))' = {1 \over {2x}} + \ln x \\
	M &= \sup_{z \in [x+{1 \over 2},x+1]} (\ln S(z))' = {1 \over {2(x+1)}} + \ln(x+1) .
	\end{align*}
    Adding the inequalities \eqref{ap1} and \eqref{ap2},
    we get
    \begin{align*}
        (\ln S(x)) + (\ln S(x+1)) &\leq (2 \ln S(x+1/2)) - \mu/2 + M/2 \\
        & \leq (2 \ln S(x+1/2)) + {1 \over 2}\ln({{x+1} \over x}) \\
        & \leq (2 \ln S(x+1/2)) + {1 \over {2x}}.
    \end{align*}
    This gives us the inequality
    	\begin{displaymath}
    		S(x)S(x+1) \le e^{1/(2x)}S(x+1/2)^2
    	\end{displaymath}
    for all $x\ge1$.
\end{proof}

Next, for $\a\in\R$, define the function $g_\a:\R_{>0}\to\R_{>0}$ by
\begin{displaymath}
	g_\a(x) = \frac{4^x}{x^{(\lg x)/4}}x^\a.
\end{displaymath}
Our goal is to show that $f$ can be approximated by $g_\a$ for some choice of $\a$.
\begin{lemma}
\label{lem:ga}
    Let $\a>0$. Then for every $x\ge4^\a$ we have
    \begin{displaymath}
        e^{-1/(2\sqrt x)}\frac52g_\a(x+\frac12) \le g_\a(x)+g_\a(x+1) \le e^{1/(2\sqrt x)}\frac52g_\a(x+\frac12) .
    \end{displaymath}
\end{lemma}
\begin{proof}
	We again compute the logarithmic derivative:
	\begin{displaymath}
		\frac\dd{\dd x}\ln g_\a(x) = \frac{\a-(\lg x)/2}x+\ln4 .
	\end{displaymath}
	For $x\ge4^\a$, this derivative is at most $\ln 4$, so by the mean value theorem, 
	$$\ln g_\a(x+1)-\ln2\le\ln g_\a(x+\frac12)\le\ln g_\a(x)+\ln2 .$$ 
	Exponentiating, we get
	\begin{equation}
		\frac12g_\a(x+1) \le g_\a(x+\frac12) \le 2g_\a(x). \label{gm1}
	\end{equation}
	Next, we note that the derivative of $\ln x$ exceeds that of $\sqrt x$ for $0<x<4$ and is less for $x>4$. So, since $\ln4<\sqrt4$, we have $\ln x<\sqrt x$ for all $x>0$.  Hence
	\begin{align*}
    	\frac\dd{\dd x}\ln g_\a(x)
		&= \frac\a x-\frac{\ln x}{2x\ln2}+\ln4 \\
		&\ge \ln4-\frac{\sqrt x}{(2\ln2)x} \\
		&\ge \ln4-1/\sqrt x 
	\end{align*}
	for all $x>0$. Thus, by the mean value theorem and exponentiating again, we get
	\begin{equation}
		\frac12e^{1/(2\sqrt x)}g_\a(x+1)\ge g_\a(x+\frac12)\ge2e^{-1/(2\sqrt x)}g_\a(x). \label{gm2}
	\end{equation}
	We can now combine the inequalities \eqref{gm1} and \eqref{gm2} for $x\ge4^\a$ to get
	\begin{align*}
		e^{-1/(2\sqrt x)}\frac52g_\a(x+\frac12)
		&\le \frac12g_\a(x+\frac12)+2e^{-1/(2\sqrt x)}g_\a(x+\frac12) \\
		&\le g_\a(x)+g_\a(x+1) \\
		&\le 2e^{1/(2\sqrt x)}g_\a(x+\frac12)+\frac12g_\a(x+\frac12) \\
		&\le e^{1/(2\sqrt x)}\frac52g_\a(x+\frac12),
	\end{align*}
	which gives us both desired bounds.
\end{proof}

We now show an identity relating $g_\a$ and $S$.
\begin{lemma}
\label{lem:gaS}
	Suppose $x>0$ and $\b>0$. If $\a=\lg\b+1/4-(\lg\pi)/2$, then
	\begin{displaymath}
		\b\frac{S(2x)}{S(x)^2}g_\a(x)
		= g_\a(2x).
	\end{displaymath}
\end{lemma}
\begin{proof}
	We have
	\begin{align*}
		\b\frac{S(2x)}{S(x)^2}g_\a(x)
		&= \b\frac{\sqrt{4\pi x}(2x/e)^{2x}}{(\sqrt{2\pi x}(x/e)^x)^2}\frac{4^x}{x^{\lg x/4}}x^\a \\
		&= \b\frac{4^x}{\sqrt{\pi x}}\frac{4^x}{x^{\lg x/4}}x^\a \\
		&= \frac{2^{\lg\b+1/4}}{\pi^{1/2}x^{1/4}x^{1/4}2^{1/4}}\frac{4^{2x}}{x^{\lg x/4}}x^\a \\
		&= 2^{\lg\b+1/4-\lg\pi/2}\frac{4^{2x}}{2^{(1+\lg x)/4}x^{(1+\lg x)/4}}x^\a \\
		&= \frac{4^{2x}}{(2x)^{(\lg2x)/4}}(2x)^\a \\
		&= g_\a(2x). \qedhere
	\end{align*}
\end{proof}

\section{Bounds on \texorpdfstring{$f(n)$}{f(n)}}
\label{bounds}

In this section we obtain an estimate for $f(n)$, the size of the optimal regular expression for~$P_n$.
\begin{theorem}
\label{thm:fn-bounds}
	For all $n\ge1$ we have
	\begin{displaymath}
		0.195\frac{4^n}{n^{(\lg n)/4}}n^{5/4-(\lg\pi)/2} \le f(n) \le {1\over4}\frac{4^n}{n^{(\lg n)/4}}n^{(\lg5)-3/4-(\lg\pi)/2} .
	\end{displaymath}
	Further, when $n$ is a power of two, we get the following upper bound, matching the general lower bound.
	\begin{displaymath}
		f(n) \le {1 \over 4}\frac{4^n}{n^{\lg n/4}}n^{5/4-(\lg\pi)/2} .
	\end{displaymath}
\end{theorem}

\begin{proof}
	Recall the Stirling approximation
	\begin{equation}
		e^{1/(12n+1)}S(n)\le n!\le e^{1/12n}S(n); \label{eqn:stirling}
	\end{equation}
	see \cite{Robbins:1955}.
	Now suppose that $f(n)\le r_ng_\a(n)$ and $f(n+1)\le r_{n+1}g_\a(n+1)$, where $n\ge\max\{1,4^\a\}$, for some non-decreasing function $r:\N\to\R_{>0}$.  Then by combining Lemma~\ref{lem:sa}, Lemma~\ref{lem:ga}, and equation~\eqref{eqn:stirling}, we get
	\begin{align*}
		f(2n+1)
		&= \binom{2n+1}n(f(n)+f(n+1)) \\
		&\le e^{\frac1{12(2n+1)}-\frac1{12n+1}-\frac1{12(n+1)+1}}\frac{S(2n+1)}{S(n)S(n+1)}(r_ng_\a(n)+r_{n+1}g_\a(n+1)) \\
		&\le \frac52r_{n+1}e^{1/(2\sqrt n)}\frac{S(2n+1)}{S(n+1/2)^2}g_\a(n+1/2) 
	\end{align*}
	and
	\begin{align*}
		f(2n)
		&= \binom{2n}n(f(n)+f(n)) \\
		&\le 2r_ne^{\frac1{12(2n)}-2\frac1{12n+1}}\frac{S(2n)}{S(n)^2}g_\a(n) \\
		&\le 2r_n\frac{S(2n)}{S(n)^2}g_\a(n).
	\end{align*}
	
	For the case where $n$ is a power of two only, we use $\b=2$ and $r_n=C$, we set 
	$$\a=\lg\b+1/4-(\lg\pi)/2=5/4-(\lg\pi)/2,$$ so $\a>0$ and $4^\a<2$. Now Lemma~\ref{lem:gaS} gives us the identity $2\frac{S(2x)}{S(x)^2}g_\a(x)=g_\a(2x)$. Then by induction we have $f(n)\le Cg_\a(n)$ for all $n\ge1$, where $C$ is any constant that satisfies this bound for $n<4$. In particular, $C={1 \over 4}$ works, so we have $f(n)\le {1 \over 4}g_{5/4-\lg\pi/2}(n)$ for all $n\ge1$ that are powers of two.

	Next, for general $n$, we use $\b=5/2$ and $r_n=Ce^{-\frac{\sqrt5}{(4-\sqrt{10})\sqrt n}}$, we set $$\a=\lg\b+1/4-\lg\pi/2=\lg5-3/4-\lg\pi/2,$$ so $\a>0$ and $4^\a<4$. Now Lemma~\ref{lem:gaS} gives us the identity $\frac52\frac{S(2x)}{S(x)^2}g_\a(x)=g_\a(2x)$. For $n\ge4$, we get
	\begin{align*}
		r_{n+1}e^{1/(2\sqrt n)}
		&= Ce^{\frac1{2\sqrt n}-\frac{\sqrt5}{(4-\sqrt{10})\sqrt{n+1}}} \\
		&\le C(e^{\frac1{\sqrt n}})^{\frac12-\frac{\sqrt5}{4-\sqrt{10}}\frac{\sqrt4}{\sqrt5}} \\
		&= C(e^{\frac1{\sqrt n}})^{-\frac{\sqrt{10}}{2(4-\sqrt{10})}} \\
		&= Ce^{-\frac{\sqrt5}{(4-\sqrt{10})\sqrt{2n}}} \\
		&\le r_{2n+1}.
	\end{align*}
	Easily, we also get $2r_n\le\frac52r_{2n}$. Thus, by induction we have $f(n)\le r_ng_\a(n)$ for all $n\ge12$, where $C$ is chosen to make this work for $12\le n<24$. In particular, $C={1\over4}$ works again. Further, since we have $r_n<C$ for all $n\ge1$, we also have $f(n)\le{1\over4}g_{\lg5-3/4-\lg\pi/2}(n)$ for all $n\ge12$. Finally, we check manually that this last inequality holds for $1\le n<12$ too, and thus for all $n\ge1$.

	All that remains is the lower bound. We get similar recurrences, supposing $f(n)\ge r_ng_\a(n)$ and $f(n+1)\ge r_{n+1}g_\a(n+1)$, where $n\ge\max\{1,4^\a\}$ for some non-increasing $r:\N\to\R_{>0}$.
	Then by a similar argument as for the upper bounds, we have
	\begin{align*}
		f(2n+1)
		&= \binom{2n+1}n(f(n)+f(n+1)) \\
		&\ge \frac52r_ne^{\frac1{12(2n+1)+1}-\frac1{12n}-\frac1{12(n+1)}}e^{-1/(2\sqrt n)}e^{-1/(2n)}\frac{S(2n+1)}{S(n+1/2)^2}g_\a(n+1/2) \\
		&\ge \frac52r_ne^{-1/(2\sqrt n)-2/(3n)}\frac{S(2n+1)}{S(n+1/2)^2}g_\a(n+1/2)
	\end{align*}
	and
	\begin{align*}
		f(2n)
		&= \binom{2n}n(f(n)+f(n)) \\
		&\ge 2r_ne^{\frac1{24n+1}-\frac2{12n}}\frac{S(2n)}{S(n)^2}g_\a(n) \\
		&\ge 2r_ne^{-1/(6n)}\frac{S(2n)}{S(n)^2}g_\a(n).
	\end{align*}
	This time, we set $\b=2$ with $r_n=Ce^{1/3n}$ (indeed non-increasing), and $\a=5/4-\lg\pi/2$, noting $4^\a<4$. Now for $n=9$, we have $\ln\frac54\ge\frac29=\frac1{2\sqrt n}+\frac1{2n}$. Since the right-hand side of this inequality is non-increasing in $n$, the bound holds for all $n\ge9$. This implies
	\begin{align*}
		\frac52e^{-1/(2\sqrt n)-2/(3n)}r_n
		&\ge \frac52e^{-1/(6n)-\ln(5/4)}r_n \\
		&= 2r_{2n} \\
		&\ge 2r_{2n+1}.
	\end{align*}
	Further, $2r_ne^{-1/6n}=2r_{2n}$, so by induction we have $f(n)\ge r_ng_\a(n)$ for all $n\ge17$, where $C$ is chosen to satisfy this bound for $17\le n<34$. In particular, $C=0.195$ works. Since $r_n>C$ for all $n\ge17$, we also have $f(n)\ge0.195g_{5/4-\lg\pi/2}(n)$ for all $n\ge17$. Finally, we check manually that this works for all $1\le n<17$ too, and thus for all $n\ge1$.
\end{proof}

\section{Optimality revisited}
\label{sec:opt-choice}

We now give a simple lower bound on the growth of $f$.
\begin{lemma}
	\label{lem:3fn}
	We have $f(n+1)\ge3f(n)$ for all $n \geq 1$.
\end{lemma}
\begin{proof}
	We prove this by induction on $n$. 
	It is easy to verify 
	 the base case $f(2)=4\ge3=3f(1)$.  Otherwise $n > 1$. Suppose the desired inequality holds for all smaller values of $n$. If $n\ge2$ is odd, then let $1\le m<n$ satisfy $2m+1=n$. Then
	\begin{align*}
		f(n+1)
		&= f(2m+2) \\
		&= 2\binom{2m+2}{m+1}f(m+1) \\
		&= 2\frac{2m+2}{m+1}\binom{2m+1}mf(m+1) \\
		&= \binom{2m+1}m(f(m+1)+3f(m+1)) \\
		&\ge \binom{2m+1}m(3f(m)+3f(m+1)) \\
		&= 3f(2m+1) \\
		&= 3f(n) .
		\end{align*}
    Otherwise, if $n\ge2$ is even, then let $1\le m<n$ satisfy $2m=n$. We note that $4m+2\ge 3m+3$, so $2\frac{2m+1}{m+1}\ge3$. We then have
\begin{align*}
		f(n+1)
		&= f(2m+1) \\
		&= \binom{2m+1}m(f(m)+f(m+1)) \\
		&= \frac{2m+1}{m+1}\binom{2m}m(f(m)+f(m+1)) \\
		&\ge \frac{2m+1}{m+1}\binom{2m}m(f(m)+3f(m)) \\
		&= 2\frac{2m+1}{m+1}f(2m) \\
		&\ge 3f(2m) \\
		&= 3f(n) . \qedhere
	\end{align*}
\end{proof}

Armed with this inequality and the bounds given by Theorem~\ref{thm:fn-bounds} of Section~\ref{bounds}, we are ready to complete our proof of the optimality of $R_n$. We recall Lemma~\ref{lem:opt-choice}, which is what we have left to show:
\optchoice*
\begin{proof}
	We easily check the cases $n<12$ by hand. Let $n\ge12$ be arbitrary. It suffices to consider the cases where $k\le\fl{n/2}$, as those where $k\ge\cl{n/2}$ are symmetric. Equality for the case $k=\fl{n/2}$ is given by the definition of $f(n)$.

	Suppose that $n/6\le k<\fl{n/2}$.  Then $n > 9$, so $3n-3 > 2n+6$.  Hence
we get $$\frac{n/2-1/2}{n/2+3/2} > 2/3$$
and so
\begin{equation}
	\frac{\fl{n/2}}{\cl{n/2}+1} > 2/3. \label{medk-1}
\end{equation}
Also, from $k \ge n/6$ we get
$3k+3 \ge \cl{n/2}+2$, and so
\begin{equation}
	\frac{k+1}{\cl{n/2}+2} \ge 1/3. \label{medk-2}
\end{equation}
Then
\begin{align*}
	&\binom nk(f(k)+f(n-k)) \\
		&\ge \binom nkf(n-k) \\
		&\ge \binom nk \, 3^{\fl{n/2}-k} \, f(\cl{n/2}) & \text{(by Lemma~\ref{lem:3fn})} \\
		&\ge 3^{\fl{n/2}-k} \, \binom nk \, \frac{f(\fl{n/2})+f(\cl{n/2})}2 \\[.1in]
		&= \frac{3^{\fl{n/2}-k}}2 \, \prod_{k\le j<\fl{n/2}}\frac{j+1}{n-j} \, \binom n{\fl{n/2}}\,\lef(f(\fl{n/2})+f(\cl{n/2})\rig) \\[.1in]
		&= \frac{3^{\fl{n/2}-k}}2 \,  \frac{\prod_{k<j\le\fl{n/2}}\, j}{\prod_{\cl{n/2}<j\le n-k} \, j} \, f(n) \\[.1in]
		&= \frac{3^{\fl{n/2}-k}}2 \, \frac{\fl{n/2}}{\cl{n/2}+1} \, \frac{\prod_{k+1\le j\le\fl{n/2}-1}\, j}{\prod_{\cl{n/2}+2\le j\le n-k} \, j} \, f(n) \\[.1in]
		&> \frac{3^{\fl{n/2}-k}}2 \, (2/3) \, \prod_{1\le j\le\fl{n/2}-k-1}\frac{k+j}{\cl{n/2}+1+j} \, f(n) & \text{(by \eqref{medk-1})} \\
		&\ge 3^{\fl{n/2}-k-1}\,\lef(\frac{k+1}{\cl{n/2}+2}\rig)^{\fl{n/2}-k-1}\,f(n) \\
		&\ge 3^{\fl{n/2}-k-1}\,(1/3)^{\fl{n/2}-k-1}\,f(n) & \text{(by \eqref{medk-2})} \\
		&= f(n).
	\end{align*}
	Next, suppose $1\le k<n/6$.  Then
$4k < n-n/3$ and so
$4 < \frac{n-1}k$.  Then if $k>1$,
\begin{align*}
		\binom nk
		&= n\prod_{2\le j\le k}\frac{n-1-k+j}j \\
		&\ge n\lef(\frac{n-1}k\rig)^{k-1} \\
		&\ge n4^{k-1}.
	\end{align*}
	We note also that $\binom n1=n4^0$, so we have $\binom nk\ge n4^{k-1}$ for all $1\le k<n/6$. We note from our proof of Lemma~\ref{lem:ga} that the derivative of $\ln g_\a(x)$ is at most $\ln4$ for $x\ge4^\a$. In particular, for $\a=5/4-\lg\pi/2$, this derivative is at most $\ln4$ for all $x\ge2$, and so $4^kg_\a(n-k)\ge g_\a(n)$ here (as $n-k>5n/6\ge10$).
	Thus
	\begin{align*}
	& \binom nk(f(k)+f(n-k)) \\
		&\ge \binom nkf(n-k) \\
		&\ge n4^{k-1}\lef(0.195g_{5/4-\lg\pi/2}(n-k)\rig) & \text{(by Theorem~\ref{thm:fn-bounds})} \\
		&\ge n4^{k-1}\lef(0.195\frac{g_{5/4-\lg\pi/2}(n)}{4^k}\rig) \\
		&= 0.195n{1\over4}g_{5/4-\lg\pi/2}(n) \\
		&= 0.195n{1\over4}n^{2-\lg5}g_{\lg5-3/4-\lg\pi/2}(n) \\
		&\ge 0.195n^{3-\lg5}f(n) & \text{(by Theorem~\ref{thm:fn-bounds})} \\
		&> f(n) & \text{(since $n\ge12>0.195^{-\frac1{3-\lg5}}$)} . & \qedhere
	\end{align*}
\end{proof}

\bibliographystyle{new}
\bibliography{perm}

\begin{thebibliography}{10}

\bibitem{Asveld:2006}
P.~R.~J. Asveld, Generating all permutations by context-free grammars in
  {Chomsky} normal form, {\em Theoret. Comput. Sci.} {\bf 354} (2006),
  118--130.

\bibitem{Asveld:2008}
P.~R.~J. Asveld, Generating all permutations by context-free grammars in
  {Greibach} normal form, {\em Theoret. Comput. Sci.} {\bf 409} (2008),
  565--577.

\bibitem{Chistikov&Ivan&Lubiw&Shallit:2017}
D.~Chistikov, S.~Ivan, A.~Lubiw, and J.~Shallit, Fractional coverings, greedy
  coverings, and rectifier networks.
\newblock In H.~Vollmer and B.~Vall{\'e}e, editors, {\em 34th Symposium on
  Theoretical Aspects of Computer Science (STACS 2017)}, Vol.~66 of {\em
  LIPIcs}, pp.  23:1--23:14. Schloss Dagstuhl --- Leibniz-Zentrum f{\"u}r
  Informatik, 2017.

\bibitem{Ehrenfeucht&Zeiger:1976}
A.~Ehrenfeucht and P.~Zeiger, Complexity measures for regular expressions, {\em
  J. Comput. System Sci.} {\bf 12} (1976), 134--146.

\bibitem{Ellul&Krawetz&Shallit&Wang:2005}
K.~Ellul, B.~Krawetz, J.~Shallit, and M.-w. Wang, Regular expressions: new
  results and open problems, {\em J. Autom. Lang. Combin.} {\bf 10} (2005),
  407--437.

\bibitem{Filmus:2011}
Y.~Filmus, Lower bounds for context-free grammars, {\em Info. Proc. Letters}
  {\bf 111} (2011), 895--898.

\bibitem{Gramlich&Schnitger:2007}
G.~Gramlich and G.~Schnitger, Minimizing nfa’s and regular expressions, {\em
  J. Comput. System Sci.} {\bf 73} (2007), 908--923.

\bibitem{Jensen:1906}
J.~L. W.~V. Jensen, Sur les fonctions convexes et les in{\'e}galit{\'e}s entre
  les valeurs moyennes, {\em Acta Math.} {\bf 30} (1906), 175--193.

\bibitem{Jiang&Ravikumar:1993}
T.~Jiang and B.~Ravikumar, Minimal {NFA} problems are hard, {\em SIAM J.
  Comput.} {\bf 22} (1993), 1117--1141.

\bibitem{Meyer&Stockmeyer:1972}
A.~R. Meyer and L.~J. Stockmeyer, The equivalence problem for regular
  expressions with squaring requires exponential space, In {\em Proc. 13th Ann.
  IEEE Symp. on Switching and Automata Theory}, pp.  125--129. IEEE, 1972.

\bibitem{Mousavi:2017}
H.~Mousavi, Lower bounds on regular expression size.
\newblock Preprint available at \url{https://arxiv.org/abs/1712.00811}, 2017.

\bibitem{Robbins:1955}
H.~Robbins, A remark on {Stirling's} formula, {\em Amer. Math. Monthly} {\bf
  62} (1955), 26--29.

\bibitem{Sloane:2018}
N.~J.~A. Sloane et~al., The on-line encyclopedia of integer sequences.
\newblock Available at \url{https://oeis.org}, 2018.

\end{thebibliography}


\end{document}